%% file: arxiv.tex
\documentclass{article}

\usepackage{lineno}
\usepackage[pagebackref=true]{hyperref}

\usepackage{amsmath}
\usepackage{tabularx}
\usepackage{amsthm}
\usepackage{authblk}
\usepackage[numbers, sort, comma, square]{natbib}
\usepackage{cleveref}
\usepackage{nicefrac}
\usepackage{amssymb}
\usepackage[colorinlistoftodos,prependcaption,textsize=footnotesize]{todonotes}

\newtheorem{proposition}{Proposition}
\newtheorem{corollary}{Corollary}

\crefname{theorem}{Theorem}{Theorems}
\crefname{lemma}{Lemma}{Lemmata}
\crefname{observation}{Observation}{Observations}
\crefname{section}{Section}{Sections}
\crefname{figure}{Figure}{Figures}
\crefname{algorithm}{Algorithm}{Algorithms}
\crefname{proposition}{Proposition}{Propositions}

\usepackage[algo2e,ruled,vlined,linesnumbered]{algorithm2e}
\SetEndCharOfAlgoLine{.}

\usepackage[shortlabels]{enumitem}

\newcommand{\G}{\mathcal G}

\newcommand{\TE}{\mathcal E}

\newcommand{\yes}{\emph{yes}}

\newcommand{\lifetime}{\tau}
\newcommand{\TG}{\mathcal G}
\newcommand{\TGcompact}{\TG = (V,(E_t)_{t=1}^\lifetime)}

\newenvironment{nscenter}
 {\parskip=5pt\par\nopagebreak\centering}
 {\par\noindent\ignorespacesafterend}

\newcommand{\problemdef}[3]{
\begin{center}   
    \fbox{~\begin{minipage}{.9\textwidth}
      \vspace{2pt} 
     
      \noindent
      \normalsize\textsc{#1}
      
      \vspace{4pt}
      \setlength{\tabcolsep}{3pt}
      \renewcommand{\arraystretch}{1.0}
      \begin{tabularx}{\textwidth}{@{}lX@{}}
	\normalsize\textbf{Input:} 	& \normalsize#2 \\
	\normalsize\textbf{Question:} 	& \normalsize#3
      \end{tabularx}
    \end{minipage}}
    \end{center}
}

\usepackage{tikz}
\usepackage{subcaption}
\usetikzlibrary{arrows}
\usetikzlibrary{shapes,snakes}

\tikzset{,fill=red
    none/.style={color=white},
    vertexset/.style={circle,draw,thick,inner sep=3pt, },
    vertexx/.style={circle,draw,thick,inner sep=2pt},
    vertex/.style={circle,draw,fill=black, inner sep= 1.5pt},
    evertex/.style={circle,draw,fill=black, inner sep= 2pt},
    node/.style={color=black,circle,draw,dotted,very thick, inner sep= 1pt,preaction={ draw,solid,blue,-, line width=0.8pt, double=blue, double distance=2\pgflinewidth, } },
    arrow/.style={->,> = latex',shorten >= 5pt,shorten <= 5pt},
    edge/.style={-,> = latex',very thick},
    edgeorange/.style={-,> = latex',preaction={ draw,orange,-, line width=0.8pt, double=orange, double distance=2\pgflinewidth, }},
    dedge/.style={-,> = latex',dashed},
    edgedashed/.style={-,> = latex',dashed,,thick,preaction={draw,blue,-, double=blue, line width=0.8pt, double distance=2\pgflinewidth, }},
    edgedasheddot/.style={-,> = latex',dash dot dot,thick,preaction={ draw,orange,-, line width=0.8pt, double=orange, double distance=2\pgflinewidth, }},
    edgedot/.style={-,> = latex',dotted,preaction={draw,blue,-, double=blue, line width=0.8pt, double distance=2\pgflinewidth, }},
    eedge/.style={-,dashed,preaction={ draw,blue,-, double=blue, double distance=2\pgflinewidth, }},
    redge/.style={-,ultra thick,preaction={draw,blue,-, double=blue, double distance=2\pgflinewidth, }},
    redgedot/.style={dotted,thick,preaction={draw,blue,-, double=blue, double distance=2\pgflinewidth, }},
    ora/.style={preaction={ draw,orange,-, line width=0.8pt, double=orange, double distance=2\pgflinewidth, }},
    blu/.style={preaction={ draw,blue,-, line width=0.8pt, double=blue, double distance=2\pgflinewidth, }},
    bludashed/.style={preaction={ draw,dashed,blue,-, line width=0.8pt, double=blue, double distance=2\pgflinewidth, }},
    gre/.style={preaction={ draw,green,-, line width=0.8pt, double=green, double distance=2\pgflinewidth, }},
    fillblue/.style={preaction={draw,fill,blue!50,-, double=blue!80, double distance=2\pgflinewidth, }},
    orablu/.style={preaction={ draw,orange,-, line width=0.8pt, double=blue, double distance=2\pgflinewidth, }},
	set/.style={draw,ellipse,inner sep=0pt,align=center},
    bluo/.style={preaction={ draw,blue!15,-, line width=0.8pt, double=blue!15, double distance=2\pgflinewidth, }},
    orao/.style={preaction={ draw,orange!15,-, line width=0.8pt, double=orange!15, double distance=2\pgflinewidth, }},
    shaded/.style={line width=0, fill=black!10, color=black!10},
}

\theoremstyle{plain}

\newtheorem{theorem}{Theorem}
\newtheorem{lemma}{Lemma}
\theoremstyle{definition}
\newtheorem{definition}{Definition}

\newtheorem{algorithm}{Algorithm}

\crefname{definition}{Definition}{Definitions}

\newcommand{\dvcn}{$\Delta$-vertex cover number}
\usepackage{etoolbox}

\title{A Faster Parameterized Algorithm for Temporal Matching} %
\date{}
\author{Philipp Zschoche}
\affil{\small Technische Universit\"at Berlin, Algorithmics and Computational~Complexity,
		Berlin, Germany\\ 
  \texttt{zschoche@tu-berlin.de}}

\begin{document}
\maketitle
\begin{abstract}
		\input{abstract}
\end{abstract}

\input{text}

\bibliographystyle{plainnat}
\bibliography{ref-short}

\end{document}

%% file: abstract.tex
A temporal graph is a sequence of graphs (called layers) over the same vertex set---describing a graph topology which is subject to discrete changes over time.
A~$\Delta$-temporal matching $M$ is a set of time edges $(e,t)$ (an edge $e$ paired up with a point in time $t$) 
such that for all distinct time edges $(e,t),(e',t') \in M$ we have that $e$ and $e'$ do not share an endpoint, or 
the time-labels $t$ and $t'$ are at least~$\Delta$ time units apart.
Mertzios et al.~[STACS~'20] provided a $2^{O(\Delta\nu)}\cdot |\TG|^{O(1)}$-time algorithm
to compute the maximum size of a $\Delta$-temporal matching in a temporal graph $\TG$, where $|\TG|$ denotes the size of $\TG$,
and $\nu$ is the \dvcn{} of $\TG$.
The \dvcn{} is the minimum number $\nu$ 
such that the classical vertex cover number of the union of any $\Delta$ consecutive layers of the temporal graph is upper-bounded by $\nu$.
We show an improved algorithm to compute a $\Delta$-temporal matching of maximum size with a 
running time of $\Delta^{O(\nu)}\cdot |\TG|$ and hence provide an exponential speedup in terms of $\Delta$.

%% file: text.tex
\section{Introduction}
\label{sec:intro}
Matchings are one of the most fundamental and best studied notions in graph theory, see \citet{lovasz2009matching} and \citet{schrijver2003combinatorial} for an overview.
Recently, \citet{baste2018temporal} and \citet{MMNZZ} studied matchings in temporal graphs. 
A \emph{temporal graph} $\TGcompact$ consists of a set $V$ of vertices and an ordered list of $\lifetime$ 
edge sets $E_1,E_2,\dots,E_\lifetime$.
A tuple $(e,t)$ is a \emph{time edge} of $\TG$ if $e \in E_t,t \in \{ 1,2,\dots,\tau \}$.
Two \emph{time edges} $(e,t)$ and $(e',t')$ are \emph{$\Delta$-independent} 
whenever the edges $e,e'$ do not share an endpoint or 
their time labels $t,t'$ are at least~$\Delta$ time units apart 
from each other (that is, $|t -t'| \geq \Delta$).\footnote{Throughout the paper, $\Delta$ always refers to this number, 
and never to the maximum degree of a static graph (which is another common use of $\Delta$).}
A~\emph{$\Delta$-temporal matching} $M$ of a temporal graph $\G$ is a set of
time edges of $\G$ which are pairwise $\Delta$-independent.
This leads naturally to the following decision problem, introduced by \citet{MMNZZ}.
\problemdef{Temporal Matching}
{A temporal graph $\TGcompact$ and integers $k,\Delta \in \mathbb N$.}
{Is there a size-$k$ $\Delta$-temporal matching in $\TG$?}
Without loss of generality, we assume that $\Delta \leq \tau$, because otherwise \textsc{Temporal Matching} can be solved using any maximum matching algorithm.
While \textsc{Temporal Matching} is polynomial-time solvable if the temporal graph has $\tau \leq 2$ layers,
it becomes NP-hard, even if $\tau=3$ and $\Delta=2$ \cite{MMNZZ}.
Driven by this NP-hardness,
\citet{MMNZZ} showed an FPT-algorithm for \textsc{Temporal Matching},
when parameterized by $\Delta$ and the maximum matching size of 
the \emph{underlying graph} $G_{\downarrow}(\TG) :=(V,\bigcup_{i=1}^\tau E_i)$ of 
the input temporal graph~$\TGcompact$.
On a historical note, one has to mention that \citet{baste2018temporalctw} 
introduced temporal matchings in a slightly different way.
The main difference to the model of \citet{MMNZZ} which we also adopt here is that the model of \citet{baste2018temporal}
requires edges to exist in at least $\Delta$ consecutive time steps in order for them to be eligible for a matching.
However, with little preprocessing an instance of the model of \citet{baste2018temporal} can be reduced to our model and 
the algorithmic ideas presented by \citet{baste2018temporal} apply as well.
Notably, there is also the related problem 
\textsc{Multistage Matchings}: 
this is a radically different way to lift 
the notion of matchings into the temporal setting.
Here, we are given a temporal graph $\TGcompact$
and we want to find a perfect (or maximum) matching for each \emph{layer} $(V,E_i)$ such that the symmetric differences of matchings for consecutive layers are small \cite{gupta2014changing,chimani2020approximating,heeger2019multistage,bampis2018multistage}.

In this paper, we consider the vertex cover number\footnote{That is, the minimum number of vertices needed to cover all edges of a graph.}
to \emph{measure} the \emph{width} of local sections (that is, $\Delta$ consecutive layers) in temporal graphs.
We call this the~\emph{\dvcn} of a temporal graph~$\TGcompact$.
Intuitively, this is the minimum number $\nu$ of vertices which we need to hit (or cover) 
all edges in any $\Delta$ consecutive layers, also called $\Delta$-window, of the temporal graph.
Note that we do not need to use the same set of vertices for all $\Delta$-windows 
and that there are temporal graphs where the vertex cover number of the underlying graph is larger than the \dvcn{}.
Formally, the $\Delta$-vertex cover number of $\TG$ is the minimum number $\nu$ such that 
for all~$i \in [\tau-\Delta+1]$ the vertex cover number of $(V,\bigcup_{t=i}^{i+\Delta-1} E_t)$ is at most~$\nu$.
It is NP-hard to decide whether the \dvcn{} of a temporal graph is at most some given value, 
because this becomes the NP-hard \textsc{Vertex Cover} problem \cite{karp1972reducibility} if $\lifetime=\Delta=1$.
However, one can show fixed-parameter tractability when parametrizing by the 
\dvcn{} using the folklore search-tree algorithm \cite{DF13} for the classical vertex cover number on each $\Delta$-window of the temporal graph.
Observe that the \dvcn{} can be smaller but not larger than the smallest sliding $\Delta$-window vertex cover, see \citet{akrida2020temporal} for details.
Note that also \textsc{Vertex Cover} has been studied in the multistage setting
\cite{DBLP:conf/iwpec/FluschnikNRZ19}.

\citet{MMNZZ} analyzed the exponential explosion in the running time of their algorithm for \textsc{Temporal Matching}
with respect to the parameters~$\Delta$ 
and the maximum matching size of the underlying graph.
It is easy to check that the running time is also upper-bounded by $2^{O(\nu\Delta)}\cdot |\TG|^{O(1)}$, 
where $\nu$ is the \dvcn{} of the input temporal graph $\TG$.
We replace the parameter maximum matching size of the underlying graph by the \dvcn{},
because the narrowed view on $\Delta$-windows makes it potentially smaller 
than its counterpart on the underlying graph
and the vertex cover number is a better established parameter in the parameterized complexity literature 
than the maximum matching size even though they are only a constant factor away from each other.

This paper contributes an improved algorithm for \textsc{Temporal Matching}
with a running time of $\Delta^{O(\nu)}\cdot |\TG|$.
Hence, this is an exponential speedup in terms of $\Delta$ compared to the algorithm of \citet{MMNZZ}.
While our algorithm closely follows the strategy of the algorithm of \citet{MMNZZ},
the speedup results from a refined application of representative sets.
Before we describe the details in \cref{sec:algo}, 
we introduce further basic notations in the next section.

\section{Preliminaries}

We denote by $\log(x)$ the ceiling of the binary logarithm of $x$ ($\lceil \log_2(x)\rceil$).
A $p$-family is a family of sets where each set has size $p$.
We refer to a set of consecutive natural numbers $[i,j] := \{ k \in \mathbb N \mid i \leq k \leq j\}$ for some $i,j \in \mathbb N$ as an \emph{interval}.
If~$i=1$, then we denote $[i,j]$ simply by $[j]$.
The \emph{neighborhood} of a vertex~$v$ and a vertex set $X$ in a graph $G=(V,E)$ is 
denoted by $N_G(v) := \{ u \in V \mid \{v,u\} \in E \}$ and $N_G(X) := \left(\bigcup_{v \in X} N_G(v)\right) \setminus X$, respectively.

The \emph{lifetime} of a temporal graph $\TGcompact$ is~$\tau$.
The \emph{size} of a temporal graph $\TGcompact$ is~$|\TG| := |V|+\sum_{t=1}^{\lifetime}\max\{{1,|E_{t}|\}}$. 
Furthermore, in accordance with the literature~\cite{CHMZ20,FMNRZ20,wu2016efficient,zschocheFMN18}, 
we assume that the lists of labels are given in ascending order. 
The \emph{set of time edges} $\TE(\TG)$ of a temporal graph~$\TGcompact$ is defined as $\{ (e,t) \mid e \in E_t \}$.
A pair $(v,t)$ is a \emph{vertex appearance} in a temporal graph $\TGcompact$ of $v$ at time $t$ if $v \in V$ and $t \in [\lifetime]$.
A time edge $(e,t)$ \emph{$\Delta$-blocks} a vertex appearance~$(v,t')$ 
(or $(v,t')$ is \emph{$\Delta$-blocked} by~$(e,t)$) 
if $v \in e$ and $|t - t'| \leq \Delta -1$.
For a time edge set $\mathcal E$ and integers~$a$ and~$b$, 
we denote by $\mathcal E[a,b] := \{ (e,t) \in \mathcal E \mid a \leq t \leq b \}$ the subset of $\mathcal E$ between the time steps $a$ and $b$.
Analogously, for a temporal graph $\TGcompact$ we denote by~$\TG[a,b]$ the temporal graph on the vertex set $V$ 
with the time edge set $\TE(\TG)[a,b]$.

A \emph{parameterized problem} is a language $L\subseteq \Sigma ^{*}\times \mathbb {N}$, where $\Sigma$ is a finite alphabet. 
The second component is called the parameter of the problem.
A parameterized problem $L$ is \emph{fixed-parameter tractable} 
if we can decide in $f(k)\cdot |x|^{O(1)}$ time 
whether a given instance $(x,k)$ is in $L$, where $f$ is an arbitrary function depending only on $k$.
An algorithm is an FPT-algorithm for parameter~$k$ if its running time is upper-bounded by $f(k)\cdot n^{O(1)}$, 
where $n$ is the input size and $f$ is a computable function depending only on $k$.

\input{algo}

\section{Conclusion}
While we could decrease the running time to solve \textsc{Temporal Matching} exponentially in terms of $\Delta$ compared 
to the algorithm of \citet{MMNZZ}, we left open 
whether 
in \cref{thm:fpt-for-vc-delta}
we can get rid of the running time dependence on $\Delta$.

\subparagraph{Acknowledgements.} The author wishes to 
thank Rolf Niedermeier and anonymous reviewers for their useful comments on the
manuscript.

%% file: algo.tex
\section{The Algorithm}
\label{sec:algo}

\citet{MMNZZ} provided a $2^{O(\Delta\nu)} |\G|^{O(1)}$-time algorithm for \textsc{Temporal Matching}. 
We now develop an improved algorithm which runs in $\Delta^{O(\nu)} |\G|$.
Formally, we show the following.
\begin{theorem}
	\label{thm:fpt-for-vc-delta}
	\textsc{Temporal Matching} can be solved 
	in $\Delta^{O(\nu)}\cdot |\G|$ time, where $\nu$ is the \dvcn{} of $\TG$.
\end{theorem}
The proof of \cref{thm:fpt-for-vc-delta} is deferred to the end of the section.
Formally, we solve the decision variant of \textsc{Temporal Matching} as it is defined in \cref{sec:intro}.
However, the algorithm actually computes the maximum size of a $\Delta$-temporal matching in a temporal graph and
with a straight-forward adjustment the algorithm can also output a $\Delta$-temporal matching of maximum size.

The algorithm behind \cref{thm:fpt-for-vc-delta} works in the same three major steps as the algorithm presented by \citet{MMNZZ}: 
\begin{enumerate}
\item\label{step1} Divide the temporal graph into disjoint $\Delta$-windows.
\item\label{step2} For each of these $\Delta$-windows compute a small family of $\Delta$-temporal matchings.
\item\label{step3} Based on the families of the Step \ref{step2},
		by dynamic programming compute
		the maximum size of a $\Delta$-temporal matching for the whole temporal graph.
\end{enumerate}
While Step~\ref{step1} is trivial and Step~\ref{step3} is similar to the algorithm of \citet{MMNZZ},
Step~\ref{step2} is where we provide new ideas leading to an improved overall running time.
Compared to the Step~\ref{step2} of \citet{MMNZZ}, we compute smaller families of $\Delta$-temporal matchings faster. 
In the next two subsections, we describe Step~\ref{step2} and Step~\ref{step3} in detail.
Afterwards, we put everything together and prove \cref{thm:fpt-for-vc-delta}.

\subsection{Step \ref{step2}: Families of $d$-complete $\Delta$-temporal matchings.}
In a nutshell, the core of Step \ref{step2} consists of an iterative computation of
a small (bounded by $\Delta^{O(\nu)}$) family of $\Delta$-temporal matchings 
for an arbitrary $\Delta$-window such that at least one of 
them is ``extendable'' to a maximum $\Delta$-temporal matching for the whole temporal graph.
	Let $\TGcompact$ be a temporal graph of lifetime $\lifetime$, 
	and let $d$ and $\Delta$ be two natural numbers such that $d\Delta \leq \tau$.
	A family $\mathcal M$ of $\Delta$-temporal matchings is~\emph{$d$-complete for~$\TG$}
	if for any $\Delta$-temporal matching $M$ of $\TG$ there is an $M' \in \mathcal M$ 
	such that $\big(M \setminus M{[\Delta(d-1)+1, \Delta d]}\big) \cup M'$ is a $\Delta$-temporal matching of $\TG$ of
	size at least~$|M|$.
	The central technical contribution of this paper is 
	a procedure to compute in~$\Delta^{O(\nu)} \cdot |\TE(\TG[\Delta(d-1)+1, \Delta d])|$ time such a \emph{$d$-complete} family $\mathcal M$ of size at most~$\Delta^{O(\nu)}$, where~$\nu$ is the \dvcn{} of $\TG$.
	Formally, we aim for the following theorem.
\begin{theorem}
		\label{lem:d-complete-family}
		Given  two natural numbers $d,\Delta$ 
		and a temporal graph $\TG$ of lifetime at least $d\Delta$ and \dvcn{} $\nu$, 
		one can compute in  $\Delta^{O(\nu)} \cdot |\TE(\TG[\Delta(d-1)+1,\Delta d])|$ time
		a family of $\Delta$-temporal matchings which is $d$-complete for $\TG$ 
		and of size at most~$\Delta^{O(\nu)}$.
\end{theorem}
To show \cref{lem:d-complete-family}, we define a binary tree where the leaves have a fixed ordering.
An order of the leaves of a rooted tree is in \emph{post order} if 
a depth-first search traversal started at the root can visit the leaves in this order.
A \emph{postfix order} of the leaves of a rooted tree is an arbitrarily chosen fixed ordering which is in post order.

\begin{definition}
		Let $\Delta \in \mathbb N$.
		A \emph{$\Delta$-postfix tree} $T$ is a rooted full binary tree of depth at most $\log(\Delta)$ with $\Delta$ many leaves $v_i,i \in [\Delta]$, such that
		$v_1,v_2,\dots,v_\Delta$ is the postfix order of the leaves.
\end{definition}
Later, the algorithm will construct a $\Delta$-postfix tree $T_v$ for each vertex $v$ in the temporal graph.
Here, each leaf in $T_v$ represents a vertex appearance of $v$. 
For example, the leaf $v_t$ represents the vertex appearance $(v,t)$.
To encode that vertex appearances of $v$ are $\Delta$-blocked until (or since) some point in time,
we will use specific ``separators''.
\begin{definition}
		Let $T$ be a $\Delta$-postfix tree rooted at $v$ with leaves $v_1,v_2,\dots,v_\Delta$ in postfix order and let $[a,b] \subset [\Delta]$.
		Then,
		an \emph{$[a,b]$-separator} of $T$ is given by~$S := N_T(\bigcup_{i \in [\Delta] \setminus [a,b]}V(P_i))$,
		where $P_i$ is the $v_i$-$v$ path in $T$.
\end{definition}
Note that this definition becomes ambiguous if $\Delta=1$. 
We will handle the case where $\Delta=1$ in a different way. 
We now consider an example, depicted in \cref{fig:fig1}, to develop some intuition on
$\Delta$-postfix trees and $[a,b]$-separators.
In \cref{fig:fig1}, we see an auxiliary graph, 
which is constructed for some $\Delta$-window in a temporal graph with only two vertices $v$ and $u$, where $\Delta=8$.
For both vertices we constructed $\Delta$-postfix trees visualized by the dashed edges.
Moreover, there is an edge between $v$ and $u$ in the first, fifth, and sixth layer of this $\Delta$-window.
This is depicted by the straight edges $\{u_i,v_i\},i \in \{ 1,5,6\}$.
In our algorithm a path from~$u$~to~$v$ (the roots of our trees) represents a time edge between~$u$ and~$v$ in the $\Delta$-window.
The $[a,b]$-separators will represent that a vertex is blocked by some time edge outside of the $\Delta$-window.
In \cref{fig:fig1}, we look at the situation where the vertexr~$u$ (vertex~$v$) is blocked 
in the first three and last two (first four and last three) 
layers of the $\Delta$-window.
To encode these blocks, we use a $[1,3]$-separator (blue/star) and a $[7,8]$-separator (orange/triangle) in the $\Delta$-postfix tree of $u$ and
a $[1,4]$-separator (green/diamond) and a $[6,8]$-separator (red/square) in the $\Delta$-postfix tree of $v$.
Note that paths from $u$ to $v$ not intersecting one of the separator vertices correspond to time edges
which can be taken into a matching even if~$u$~and~$v$ are blocked by some other time edges 
outside of the $\Delta$-window, as depicted by the gray areas in \cref{fig:fig1}.

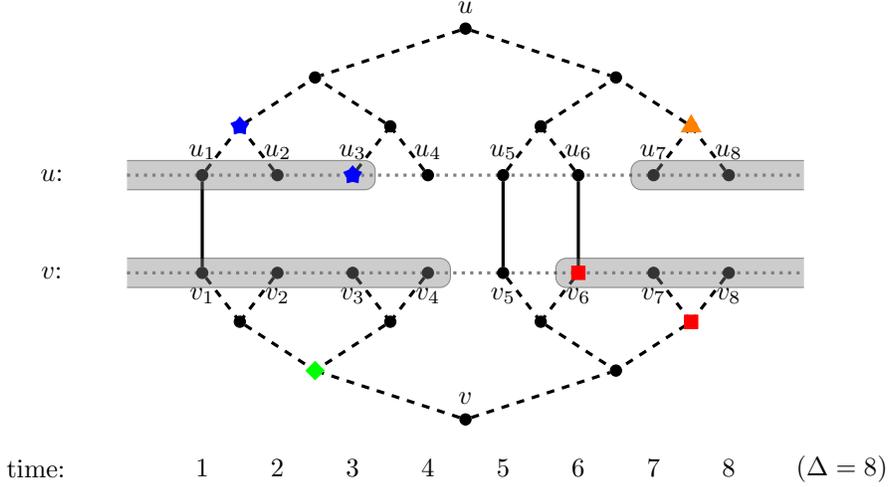
\begin{figure}
\begin{tikzpicture}[scale=1,yscale=1.3]
\foreach \i in {0,...,8} {
		\draw[edge,dotted,gray] (\i,1) -- (\i+1,1) {};
		\draw[edge,dotted,gray] (\i,2) -- (\i+1,2) {};
}
\foreach \i in {1,...,8} {
		\node[vertex,label=$u_{\i}$] (w\i) at (\i,2) {};
		\node[vertex,label=below:$v_{\i}$] (v\i) at (\i,1) {};
		\node at (\i,-1) {$\i$};
}
\foreach \i in {1,...,4} {
		\node[vertex] (w1\i) at (\i*2-0.5,2.5) {};
		\draw[edge,dashed] (\i*2-0.5,2.5) -- (\i*2-1,2) {};
		\draw[edge,dashed] (\i*2-0.5,2.5) -- (\i*2,2) {};
		\node[vertex] (v1\i) at (\i*2-0.5,0.5) {};
		\draw[edge,dashed] (\i*2-0.5,0.5) -- (\i*2-1,1) {};
		\draw[edge,dashed] (\i*2-0.5,0.5) -- (\i*2,1) {};
}
\foreach \i in {1,...,2} {
		\node[vertex] (w2\i) at (\i*4-1.5,3) {};
		\node[vertex] (v2\i) at (\i*4-1.5,0) {};
}
		\node[vertex,label=above:$u$] (w) at (4.5,3.5) {};
		\node[vertex,label=above:$v$] (v) at (4.5,-0.5) {};

		\draw[edge,dashed] (w) -- (w21) {};
		\draw[edge,dashed] (w) -- (w22) {};
		\draw[edge,dashed] (w11) -- (w21) {};
		\draw[edge,dashed] (w12) -- (w21) {};
		\draw[edge,dashed] (w13) -- (w22) {};
		\draw[edge,dashed] (w14) -- (w22) {};

		\draw[edge,dashed] (v) -- (v21) {};
		\draw[edge,dashed] (v) -- (v22) {};
		\draw[edge,dashed] (v11) -- (v21) {};
		\draw[edge,dashed] (v12) -- (v21) {};
		\draw[edge,dashed] (v13) -- (v22) {};
		\draw[edge,dashed] (v14) -- (v22) {};
		\node  at (0-1,1) {$v$:};
	 	\node  at (0-1,2) {$u$:};
		\node  at (-.21-1,-1) {time:};
		\node at (9.5,-1) {$(\Delta=8)$};
	  
      \draw[edge] (v1) -- (w1);

      \draw[edge] (v5) -- (w5);
      \draw[edge] (v6) -- (w6);

	  \draw[rounded corners,fill=gray,opacity=0.4] (9,2.15) to (6.7,2.15) to (6.7,1.85) to (9,1.85);
	  \draw[rounded corners,fill=gray,opacity=0.4] (0,2.15) to (3.3,2.15) to (3.3,1.85) to (0,1.85);
	  \draw[rounded corners,fill=gray,opacity=0.4] (0,1.15) to (4.3,1.15) to (4.3,0.85) to (0,0.85);
	  \draw[rounded corners,fill=gray,opacity=0.4] (9,1.15) to (5.7,1.15) to (5.7,0.85) to (9,0.85);

	  \node[vertex,rectangle,red,minimum size=5pt]  at (4*2-0.5,0.5) {};
	  \node[vertex,rectangle,red,minimum size=5pt]  at (3*2,1) {};

	  \node[vertex,diamond,green,minimum size=7pt]  at (2.5,-0) {};

	  \node[vertex,star,blue,minimum size=7pt]  at (1.5,2.5) {};
	  \node[vertex,star,blue,minimum size=7pt]  at (3,2) {};
	  
	  \node[vertex,orange,regular polygon,regular polygon sides=3,minimum size=7pt]  at (7.5,2.5) {};

\end{tikzpicture}
\caption{Illustration how the $\Delta$-postfix trees are used for a $\Delta$-window of a temporal graph with only two vertices $u$ and $v$. 
The edge $\{u,v\}$ is in the first, fifth, and sixth edge set of the $\Delta$-window.}
\label{fig:fig1}
\end{figure}

It is crucial for our algorithm that the $[a,b]$-separators are small in terms of~$\Delta$.
\begin{lemma}
		\label{lem:postfix-tree}
	Let $T$ be a $\Delta$-postfix tree $T$ rooted at $v$ with leaves $v_1,v_2,\dots,v_\Delta$ in postfix order,
	and let $S$ be the $[a,b]$-separator of $T$, where $[a,b] \subset [\Delta]$ with $a=1$ or $b=\Delta$.
	Then, 
	\begin{enumerate} 
	\item $|S| \leq \log(\Delta)$, and 
	\item the $v_i$-$v$ path in $T$ contains a vertex from $S$ if and only if $i \in [a,b]$.
	\end{enumerate} 
\end{lemma}
\begin{proof}
	Note that $[\Delta] \setminus [a,b]$ contains either $1$ or $\Delta$ but not both.
	We may assume without loss of generality that $\Delta \in [\Delta] \setminus [a,b]$ and hence $a=1$.
	Let $c = \min{([\Delta] \setminus [a,b])}$ and $P_c$ be the $v_c$-$v$ path in $T$.
	Observe that $S \subseteq N_T(V(P_{c}))$.
	Since (1) $P_c$ contains at most $\log(\Delta)$ vertices which are of degree at least two in $T$, and 
	(2) a vertex on $P_c$ has at most one neighbor not being in $V(P_c)$, we have that $|S| \leq \log(\Delta)$.
	The rest of the lemma follows simply by the fact that the root $v$ is in the set $\bigcup_{i \in [\Delta] \setminus [a,b]}V(P_i)$ 
	and that $\left(\bigcup_{i \in [\Delta] \setminus [a,b]}V(P_i)\right)\cap N(\bigcup_{i \in [\Delta] \setminus [a,b]}V(P_i)) = \emptyset$,
	where $P_i$ is the $v_i$-$v$ path in $T$.
\end{proof}

Our intermediate goal now is to compute a family of paths between roots of $\Delta$-postfix trees in the auxiliary graph 
such that if we are given a set of $[a,b]$-separators $S$, 
then the family shall contain a path between two roots which avoids the vertices in $S$ (if there is one). 
We will use \emph{representative families} for this.

Before we jump into the formal definition of representative families \cite{MONIEN1985239},
we build some intuition for representative families by considering an illustrative game played by Alice and Bob.
Bob has a set $U$ and a $p$-family $\mathcal S \subseteq 2^U$.
He shows $U$ and $\mathcal S$ once to Alice.
Afterwards, Bob puts $q$ elements from $U$ on the table $Y$ and asks Alice
whether there is a set in $\mathcal S$ which does not contain any element on the table $Y$.
Alice wins if and only if she can answer the question correctly.
The goal of Alice is to win this game while remembering as little as possible from $\mathcal S$.
This is represented by a set $\widehat{\mathcal S} \subseteq \mathcal S$.
Intuitively speaking, a representative family $\widehat{\mathcal S} \subseteq \mathcal S$ guarantees to
Alice that there is at least one set in~$\widehat{\mathcal S}$ 
which does not contain an element on the table $Y$ 
if there is a set in $\mathcal S$ 
which does not contain an element on the table $Y$.
Formally, we define representative families as follows.
\begin{definition}
		\label{def:repfam}
		Let $\mathcal S$ be a $p$-family
		and $\omega \colon \mathcal S \rightarrow \mathbb N$.
		A subfamily $\widehat{\mathcal S} \subseteq \mathcal S$
		is a \emph{max~$q$-representative} with respect to $\omega$
		if for each set $Y$ of size at most $q$ it holds true that
		if there is a set $X \in \mathcal S$ 
		with $X \cap Y = \emptyset$,
		then there is an $\widehat X \in \widehat{\mathcal S}$ 
		such that  $\widehat X \cap Y = \emptyset$ and $\omega(\widehat X) \geq \omega(X)$.
\end{definition}
Representative families are algorithmically useful
because there are FPT-algorithms for the parameter $p+q$ to compute $q$-representatives of $p$-families  such that the size of 
the $q$-representative only depends on $p+q$ \cite{MONIEN1985239,FLPS16}.

An algorithm of \citet{FLPS16} can be iteratively applied  
to show the following.
\begin{proposition}{\cite[Proposition 4.8]{ROP-Arxiv18}}
	\label{thm:matroid-tool}
  Let $\alpha$, $\beta$, and $\gamma$ be non-negative integers such that $r=(\alpha+\beta)\gamma \geq 1$, and let
  $\omega \colon U \to \mathbb N$~be a weight function.
  Furthermore, let $U$ be a set, $\mathcal H\subseteq 2^U$~be a $\gamma$-family of size~$t$ and let
  $\mathcal S = \{ S = \biguplus_{i=1}^\alpha H_i  
    \mid
	H_j \in \mathcal H\text{ for }j \in [\alpha]\}.
  $
  Then, we can compute a max $\beta \gamma$-representative~$\widehat{\mathcal S}$   of $\mathcal{S}$ with respect to~$\omega'$
  in $2^{O(r)}\cdot t$ time such that $|\widehat{\mathcal S}| \leq {{r}\choose {\alpha \gamma}}$, where $\omega'(X) := \sum_{x \in X} \omega(x)$.
  \end{proposition}
  Note that \citet{ROP-Arxiv18} actually showed a more general version of \cref{thm:matroid-tool}. 
  However, for the following algorithm, we only need \cref{thm:matroid-tool} (that is, Proposition 4.8 of \citet{ROP-Arxiv18}
  in the special case of a uniform matroid represented over a large enough prime field).
  
  Now, we can describe the algorithm behind \cref{lem:d-complete-family} in detail.
\begin{algorithm}[Algorithm behind \cref{lem:d-complete-family}]
		\label{const:family}
		Let $\TGcompact$ be a temporal graph of 
		lifetime $\lifetime$, and let
		$d$ and $\Delta>1$ be two natural numbers such that~$d\Delta \leq \tau$.
		Furthermore let $\TG' := \TG{[\Delta(d-1)+1, \Delta d]}$, 
		and let $\nu$ be the \dvcn{} of $\TG$.

		\begin{enumerate}[(i)]
				\item 
						For each vertex $u \in V$, 
						we construct the $\Delta$-postfix tree $T_u$ with $\Delta$ many leaves.
						These trees have pair-wise disjoint vertex sets.
		The root of $T_u$ is~$u$ and
		the leaves in postfix order are 
		$u_1,u_2,\dots,u_\Delta$.
		
		\item
		Construct a $(2\log \Delta+3)$-family $\mathcal H := \mathcal H_E \cup \mathcal H_D$ such that
			 $\mathcal H_D$ contains~$\nu$ pairwise disjoint sets of fresh vertices, and
			 $\mathcal H_E := \{ E_{(\{u,w\},t)} \mid (\{u,w\},t)$ is time edge of $\TG' \}$,
					where $E_{(\{u,w\},t)} :=$
					\begin{align*}
							 \bigcup_{y \in \{u,w\}} \left\{ x \in V(T_y) \mid x \text{ is on the $y$-$y_{t-\Delta(d-1)}$-path in } T_y \}
							\cup \{ (\{u,w\},t) \right\},
					\end{align*}
					for all time edges $(\{u,w\},t)$ of $\TG'$.
			\item Let $\omega : 2^U \to \mathbb N$ with $\omega(X) := |X \cap \TE(\TG')|$, for all $X \in 2^U$ be a weight function,
				where~$U := \bigcup_{A \in \mathcal H} A$.
		\item Compute the max $(4\nu\log{\Delta})$-representative $\widehat{\mathcal S}$ of 
		\begin{align*}
				\mathcal S := \left\{ \biguplus_{i=1}^{\nu} H_i \mid H_i \in \mathcal H, \text{ for } i \in [\nu] \right\}
		\end{align*}
				with respect to $\omega$ (using \cref{thm:matroid-tool}).
		\item Output $\mathcal M := \{ S \cap \mathcal E(\TG') \mid S \in \widehat{\mathcal S}   \}$.
		\hfill$\diamond$
\end{enumerate}
\end{algorithm}

Towards the correctness of \cref{const:family}, we observe the following.
\begin{lemma}
		\label{lem:matchings-sets}
		Let $\mathcal S$, $\TG'$, and $\nu$ be defined as in \cref{const:family} for some temporal graph $\TG$ and $d, \Delta \in \mathbb N$.
		Then,
		$M$ is a $\Delta$-temporal matching in $\TG'$ 
		if and only if there is an $S \in \mathcal S$ such that 
		$M =  \TE(\TG') \cap S$ and $\omega(S) = |M|$.
\end{lemma}
\begin{proof}
		($\Leftarrow$): Let $S \in \mathcal S$ and set $M = \TE(\TG') \cap S$.
		Clearly, $\omega(S) = |M|$ and for two distinct time edges $(e,t),(e',t') \in M$ we have that $e \cap e' = \emptyset$,
		because otherwise~$E_{(e,t)} \cap E_{(e',t')} \not= \emptyset$ and hence $S \not \in \mathcal S$.

		($\Rightarrow$): Let $M$ be a $\Delta$-temporal matching in $\TG'$.
		Since all time edges of $\TG'$ are in $\Delta$ consecutive time steps, 
		we know that for all $(e,t),(e',t') \in M$ we have that $e \cap e' = \emptyset$.
		Hence, $|M| \leq \nu$ and $E_{(e,t)} \cap E_{(e',t')} = \emptyset$.
		Thus, $S := \biguplus_{(e,t) \in M} E_{(e,t)} \uplus \biguplus_{i=1}^{\nu-|M|} D_i \in \mathcal S$ and $\omega(S) = |M|$, 
		where  $D_1,\dots,D_{\nu-|M|} \in \mathcal H_D$ are pairwise disjoint.
\end{proof}
We now show the correctness of \cref{const:family}.
\begin{lemma}
		\label{lem:d-complete-family-correct}
		Let $\mathcal M$, $\TG'$, and $\nu$ be defined as in \cref{const:family} for some temporal graph $\TG$, and $d, \Delta \in \mathbb N$.
		Then, $\mathcal M$ is a $d$-complete family of $\Delta$-temporal matchings for $\TG'$.
\end{lemma}
\begin{proof}
		By \cref{lem:matchings-sets} together with Step (iv) and (v) of \cref{const:family},
		the family~$\mathcal M$ only contains $\Delta$-temporal matchings in $\TG'$.
		To show that $\mathcal M$ is $d$-complete, 
		let $M$ be a $\Delta$-temporal matching for the whole temporal graph $\TG$.
		Then, $M' := M[\Delta(d-1)+1, \Delta d]$ is the $\Delta$-temporal matching in $\TG'$ which is included in $M$.
		If $d>1$, then let $M^- := M[\Delta(d-2)+1, \Delta (d-1)]$ or otherwise $M^- := \emptyset$.
		If $d< \nicefrac{\lifetime}{\Delta}$, then let $M^+ := M[\Delta d+1, \Delta (d+1)]$ or otherwise $M^+ := \emptyset$.
		Observe that $|M^-|,|M'|,|M^+| \leq \nu$ and 
		that vertex appearances in $\TG'$ which are incident to an arbitrary time edge
		can only be $\Delta$-blocked by time edges in $M^- \cup M^+$.
		Hence, there at most $4\nu$ vertices for which some vertex appearances in~$\TG'$ are $\Delta$-blocked by time edges in $M^- \cup M^+$.
		Let
		\begin{align*}
				B :=\ &\{ (v,[1,t-\Delta d]) \mid (e,t) \in M^-, v \in e \} \cup \\
				  &  \{ (v,[t+1 - \Delta d, \Delta]) \mid (e,t) \in M^+, v \in e \}.
		\end{align*}
		Thus, a vertex appearance $v_t$ from $\TG'$ is $\Delta$-blocked by some time edge in $M^- \cup M^+$ if and only if
		there is a $(v,[a,b]) \in B$ with $t - \Delta(d-1) \in [a, b]$.
		Now, let~$Y := \bigcup_{(v,[a,b]) \in B} S_{(v,[a,b])}$, 
		where $S_{(v,[a,b])}$ is an $[a,b]$-separator in the $\Delta$-postfix tree~$T_v$.
		Furthermore, by \cref{lem:matchings-sets}, there is an $S \in \mathcal S$ 
		such that $M' =  \TE(\TG') \cap S$ and $\omega(S) = |M'|$.

		We now show that $S \cap Y = \emptyset$. 
		Assume towards a contradiction that $S \cap Y \not= \emptyset$.
		Hence, there is an $(e,t) \in M'$ such that there is a  $u \in E_{(e,t)} \cap Y$.
		Since $u \in Y$, there is a $v \in e$ such that $u \in V(T_v)$, 
		and a $(v,[a,b]) \in B$ such that $u \in S_{(v,[a,b])}$.
		From $u \in E_{(e,t)}$ we know that $u$ is on the $v$-$v_{t-\Delta(d-1)}$ path in $T_v$.
		Hence, by \cref{lem:postfix-tree}, $t \in [a,b]$ and thus there is a time edge $(e',t') \in M^- \cup M^+$ which is not $\Delta$-independent with $(e,t)$.
		This contradicts $M$ being a $\Delta$-temporal matching. 
		Thus, $S \cap Y = \emptyset$.

		Since $S \cap Y = \emptyset$, $S \in \mathcal S$, $|Y| \leq 4\nu\log \Delta$, and
		$\widehat{\mathcal S}$ is a max $(4\nu\log \Delta)$-representative of $\mathcal S$ with respect to $\omega$,
		we know that there is an $\widehat S \in \widehat{\mathcal S}$ 
		such that~$\widehat S \cap Y = \emptyset$ and $\omega(\widehat S) \geq \omega(S)$.
		By the construction of $\mathcal M$ in \cref{const:family}, 
		and by \cref{lem:matchings-sets}, we know that there is an $\widehat M \in \mathcal M$
		such that $\widehat S \cap E(\TG') = \widehat M$ and $|\widehat M| = \omega(\widehat S) \geq \omega(S) = |M'|$.
		Hence, $|(M \setminus M') \cup \widehat M| \geq |M|$.

		We now show that $(M \setminus M') \cup \widehat M$ is a $\Delta$-temporal matching.
		Suppose not.
		Then there are time edges $(e,t) \in M^- \cup M^+$ and $(\widehat e, \widehat t) \in \widehat M$
		with $v \in e \cap \widehat e$ such that the vertex appearance $v_{\widehat t}$ is $\Delta$-blocked by $(e,t)$.
		Hence, there is a~$(v,[a,b]) \in B$ with $\widehat t - \Delta(d-1) \in [a,b]$.
		By \cref{lem:postfix-tree}, this contradicts $\widehat S \cap Y = \emptyset$.
		Hence, $(M \setminus M') \cup \widehat M$ is a $\Delta$-temporal matching and thus $\mathcal M$ is~$d$-complete.
\end{proof}
The running time of the dynamic program defined in \eqref{eq:dp} will be discussed directly in the following proof of \cref{lem:d-complete-family}.
\begin{proof}[Proof of \cref{lem:d-complete-family}]
		By \cref{lem:d-complete-family-correct}, we can use \cref{const:family} 
		to compute a $d$-complete family $\mathcal M$ of $\Delta$-temporal matchings in $\TG[\Delta(d-1)+1,\Delta d]$.
		It is easy to verify that we can compute $\mathcal H$ in $O\left((\nu + |\TE(\TG[\Delta(d-1)+1,\Delta d])|)\log \Delta \right)$ time (by ignoring isolated vertices).
		Finally, we compute $\widehat{\mathcal S}$ with \cref{thm:matroid-tool} in $2^{O(\nu \cdot \log{\Delta})}\cdot |\TE(\TG[\Delta(d-1)+1,\Delta d])|$ time,
		by setting 
		$\alpha$ to $\nu$, 
		$\beta$ to $2\nu$, and 
		$\gamma$ to $2 \log{\Delta} + 3$.
		Note that the computed family of \cref{thm:matroid-tool} is even a $(4\nu\log{\Delta}+6\nu)$-representative and 
		hence suffices as $(4\nu\log{\Delta})$-representative, see \cref{def:repfam}.
		By \cref{thm:matroid-tool} the size of $\widehat{\mathcal S}$ is at most $2^{O(\nu \cdot \log{\Delta})} = \Delta^{O(\nu)}$.
		Hence, we end up with an overall running time of~$\Delta^{O(\nu)}\cdot |\TE(\TG[\Delta(d-1)+1,\Delta d])|$.
\end{proof}

\subsection{Step \ref{step3}: The dynamic program}
In this section we describe Step \ref{step3} of the algorithm behind \cref{thm:fpt-for-vc-delta}, see \cref{sec:algo}.

Let $\TGcompact$ be a temporal graph such that $\lifetime$ is a multiple of $\Delta \in \mathbb N$.
Assume that we already computed for all $d \in [\nicefrac{\lifetime}{\Delta}]$ a family $\mathcal M_d$ of $\Delta$-temporal matchings 
which is $d$-complete for $\TG$.

For all $i \in [\nicefrac{\lifetime}{\Delta}]$ 
and $M \in \mathcal M_i$, for $i > 1$ let
\begin{equation}
		\label{eq:dp}
		\begin{split}
		&T_i[M] := \left\{ |M| + T_{i-1}[M'] \  \middle\vert \
		M'\in \mathcal M_{i-1}, M \cup M' \text{ is a $\Delta$-temporal matching }
\right\},\\
		&\text{and } T_1[M] := |M|.
		\end{split}
\end{equation}
Towards the correctness of the dynamic program specified in \eqref{eq:dp}, 
we observe the following.

\begin{lemma}
		\label{lem:dp-correctness}
		There is a $\Delta$-temporal matching of size at least $k$ in $\TG$
		if and only if
		$\max_{M \in \mathcal M_\frac{\lifetime}{\Delta}}  T_{\frac{\lifetime}{\Delta}}[M] \geq k$.
\end{lemma}
\begin{proof} 
		($\Rightarrow$):
		We show by induction over $i$ that if there is a $\Delta$-temporal matching~$M$ 
		in $\TG$,
		then there is an $M' \in \mathcal M_i$
		such that $T_{i}[M'] \geq |M[1,\Delta i]|$ and $M[1,\Delta(i-1)]\cup M' \cup M[\Delta i+1, \lifetime] $ is a $\Delta$-temporal matching
		of size at least~$|M|$.
		By~\eqref{eq:dp}, this is clearly the case for $i=1$, because $\mathcal M_1$ is $1$-complete for~$\TG$.

		For the induction step, let $i>1$ and
		assume that
		if there is a $\Delta$-temporal matching $M$ 
		in $\TG$, then
		there is an $M' \in \mathcal M_{i-1}$
		such that
		\begin{enumerate}[(i)]
				\item $T_{i-1}[M'] \geq |M[1,\Delta(i-1)]|$, and 
				\item $M[1,\Delta(i-2)]\cup M' \cup M[\Delta(i-1)+1,\tau]$ is a $\Delta$-temporal matching of size at least $|M|$.
		\end{enumerate}
		Let $M^*$ be a $\Delta$-temporal matching for $\TG$.
		By the induction hypothesis, 
		there is an $M' \in \mathcal M_{i-1}$ 
		such that $T_{i-1}[M'] \geq |M^*[1,\Delta(i-1)]|$
		and $\widehat M := M^*[1,\Delta(i-2)] \cup M' \cup M^*[\Delta(i-1)+1,\tau]$ is a $\Delta$-temporal matching of size at least $|M^*|$.
		Since $\mathcal M_i$ is $i$-complete for $\TG$,
		there is an $M'' \in \mathcal M_i$ such that $\widehat M[1,\Delta(i-1)] \cup M'' \cup \widehat M[\Delta i + 1,\tau]$
		is a $\Delta$-temporal matching of size at least $|\widehat M|\geq |M^*|$.
		By~\eqref{eq:dp}, we have that $T_i[M''] \geq |\widehat M[1,\Delta i]| \geq |M^*[1,\Delta i]|$.
		Hence, if there is a $\Delta$-temporal matching of size $k$ in $\TG$,
		then there is an~$M' \in \mathcal M_{\frac{\tau}{\Delta}}$ such that $T_{\frac{\tau}{\Delta}}[M'] \geq k$.
		
		($\Leftarrow$): 
		We show by induction over $i$ that for all $M' \in \mathcal M_i$ there is a $\Delta$-temporal matching $M$ in $\TG[1,\Delta i]$
		such that $|M|=T_i[M']$ and $M[\Delta(i-1)+1,\Delta i] = M'$.
		By~\eqref{eq:dp}, this is clearly the case for $i=1$, because $\mathcal M_1$ is $1$-complete for~$\TG$.

		For the induction step, let $i > 1$ and assume 
		that for all $M' \in \mathcal M_{i-1}$
		there is a $\Delta$-temporal matching $M$ in $\TG[1,\Delta(i-1)]$
		such that $|M| = T_{i-1}[M']$ and $M[\Delta(i-2)+1,\Delta (i-1)] = M'$.
		Let $M'' \in \mathcal M_{i}$.
		By~\eqref{eq:dp} and since $\mathcal M_{i-1}$ is $(i-1)$-complete,
		there is an $M' \in \mathcal M_{i-1}$ such that
		$M'' \cup M'$ is a $\Delta$-temporal matching and $T_i[M''] = T_{i-1}[M'] + |M''|$.	
		By the induction hypothesis, there is a $\Delta$-temporal matching $M$ in $\TG[1,\Delta(i-1)]$
		such that $|M| = T_{i-1}[M']$ and $M[\Delta(i-2)+1,\Delta (i-1)] = M'$.
		Since $M[\Delta(i-2)+1,\Delta (i-1)] = M'$ and $M' \cup M''$ is a $\Delta$-temporal matching, 
		$M \cup M''$ is a $\Delta$-temporal matching of size $T_{i-1}[M'] + |M''| = T_{i}[M'']$.
		Hence, if $\max_{M \in \mathcal M_\frac{\lifetime}{\Delta}}  T_{\frac{\lifetime}{\Delta}}[M] \geq k$,
		then there is a $\Delta$-temporal matching of size at least $k$ in $\TG$.

\end{proof} 

We now are ready to prove \cref{thm:fpt-for-vc-delta}.

\begin{proof}[Proof of \cref{thm:fpt-for-vc-delta}]
		Let $(\TG,k,\Delta)$ be an instance of \textsc{Temporal Matching}.
		We assume without loss of generality that there is no $\Delta$-window in $\TG$ which does not contain any time edge, 
		otherwise we can split $\TG$ into two parts, 
		compute the maximum size of a $\Delta$-temporal matching in each part separately,
		and check whether the sum is at least $k$.
		Furthermore, we may assume that $\Delta>1$; otherwise we solve $(\TG,k,\Delta)$ by checking 
		whether the sum of the maximum matching size for each layer is at least $k$.
		The running time for the case $\Delta=1$ can be bounded by $O(\nu|\TG|)$ because the folklore Ford–Fulkerson method
		needs at most~$\nu$ rounds to compute a maximum matching for any layer, where $\nu$ is the \dvcn{} of $\TG$.
		Moreover, we can assume without loss of generality that the lifetime $\tau$ of $\TG$ is a multiple of $\Delta$, otherwise we can add some empty layers at the end of $\TG$.

		We start by splitting $\TG$ into $\nicefrac{\tau}{\Delta}$ many $\Delta$-windows:
		for all $i \in [\nicefrac{\tau}{\Delta}]$ let~$\TG_i := \TG[\Delta(i-1)+1,\Delta i]$.
		This can be done in $O(|\TG|)$ time.
		To compute the \dvcn{}, 
		we first compute the vertex cover number of the underlying graph of $\TG_i$
		in $O(2^{\nu+1} |\TE(\TG_i)|)$ time by the folklore search-tree algorithm \cite{DF13}, for all~$i \in [\nicefrac{\tau}{\Delta}]$.
		Let $\nu'$ be the maximum vertex cover number over all underlying graphs $\TG_i$, where  $i \in [\nicefrac{\tau}{\Delta}]$.
		Note that there are only $\nu'+1$ possible values for the \dvcn{} $\nu$ of~$\TG$ since $\nu' \leq \nu \leq 2\nu'$.
		Hence, we can compute~$\nu'$ and then guess $\nu$ by branching into the $\nu'+1$ possible choices in $2^{O(\nu)}\cdot |\TG|$ time.  

		Consider the branch where we have chosen the correct \dvcn{}~$\nu$.
		We compute with \cref{lem:d-complete-family} for each $i \in [\nicefrac{\tau}{\Delta}]$ 
		a family $\mathcal M_i$ of $\Delta$-temporal matching of size $\Delta^{O(\nu)}$ which is $i$-complete for $\TG$ 
		in $\Delta^{O(\nu)} \cdot |\TE(\TG_i)|$. 
		Hence, it takes $\Delta^{O(\nu)} \cdot |\TG|$ time to compute the families~$\mathcal M_1,\dots,\mathcal M_{\frac{\tau}{\Delta}}$.

		Now we have met the preconditions to compute the dynamic program specified in~\eqref{eq:dp}.
		By \cref{lem:dp-correctness}, there is a $\Delta$-temporal matching of size at least $k$ in~$\TG$
		if and only if $\max_{M \in \mathcal M_{\frac{\tau}{\Delta}}} T_{\frac{\tau}{\Delta}}[M] \geq k$.
		Hence, if $(\TG,k,\Delta)$ is a \yes-instance,
		then at some point we pick the right choice for $\nu$ and correctly conclude that the input is a \yes-instance.
		To ensure that $(\TG,k,\Delta)$ is a \yes-instance if we say so, we perform a sanity-check by verifying 
		whether there is indeed a $\Delta$-temporal matching of size at least $k$ in our dynamic program.
		This can be done by additionally storing for the table entry $T_i[M'],i\in[\nicefrac{\tau}{\Delta}],M' \in \mathcal M_i$
		a $\Delta$-temporal matching $M$ in $\TG[1,\Delta i]$ of size $T_i[M']$ such that~$M[\Delta(i-1)+1,\Delta i] = M'$.
		Note that $\max_{M \in \mathcal M_{\frac{\tau}{\Delta}}} T_{\frac{\tau}{\Delta}}[M]$
		can be computed in $\Delta^{O(\nu)} \cdot \sum_{i=1}^{\frac{\tau}{\Delta}} |\TE(\TG[\Delta(i-1),\Delta i)]|$ time.
		Since each $\Delta$-window contains at least one time edge, we need $\Delta^{O(\nu)} \cdot |\TG|$ time 
		for each possible choice of $\nu$.
		Hence, the overall running for all choices of $\nu$ is $\Delta^{O(\nu)} \cdot |\TG|$ time.

		This completes the proof.
\end{proof}	
As already discussed in the proof of \cref{thm:fpt-for-vc-delta}, 
we can modify the dynamic program in~\eqref{eq:dp} to compute a $\Delta$-temporal matching of maximum size and not just the size.
Thus, we can solve the optimization variant of \textsc{Temporal Matching}.
\begin{corollary}
		Given a temporal graph $\TG$ and an integer $\Delta$,
		one can compute in~$\Delta^{O(\nu)}\cdot |\TG|$ time a maximum-cardinality $\Delta$-temporal matching in $\TG$, where~$\nu$ is the~\dvcn{}.
\end{corollary}